\documentclass[pdflatex,sn-mathphys-num]{sn-jnl}

\usepackage{graphicx}%
\usepackage{multirow}%
\usepackage{amsmath,amssymb,amsfonts}%
\usepackage{amsthm}%
\usepackage{mathrsfs}%
\usepackage[title]{appendix}%
\usepackage{xcolor}%
\usepackage{textcomp}%
\usepackage{manyfoot}%
\usepackage{booktabs}%
\usepackage{algorithm}%
\usepackage{algpseudocode}%
\usepackage{listings}%
\usepackage{tikz}
\usepackage{tabularx}
\usepackage{caption}
\usepackage{braket}
\theoremstyle{thmstyleone}%
\newtheorem{theorem}{Theorem}
\newtheorem{corollary}{Corollary}

\theoremstyle{thmstyletwo}%
\newtheorem{remark}{Remark}%

\theoremstyle{thmstylethree}%

\begin{document}

\title[Article Title]{Certified Lower Bounds and Efficient Estimation of Minimum Accuracy in Quantum Kernel Methods}

\author*[1]{\fnm{Demerson} \sur{N. Gonçalves}}\email{demerson.goncalves@cefet-rj.br}

\author[2]{\fnm{Tharso}  \sur{D. Fernandes}}\email{tharso.fernandes@ufes.br}

\author[3]{\fnm{Andrias } \sur{M. M. Cordeiro}}\email{andrias.cordeiro@aluno.cefet-rj.br}
\author[3]{\fnm{Pedro  } \sur{H. G. Lugao}}\email{pedro.lugao@cefet-rj.br}
\author[4]{\fnm{João  } \sur{T. Dias}}\email{joao.dias@cefet-rj.br}


\affil*[1]{Dept. of Mathematics, Federal Center for Technological Education Celso Suckow da Fonseca (CEFET-RJ), Petrópolis, RJ, Brazil}
\affil[2]{Dept. of Mathematics, Federal University of Espírito Santo (UFES), Alegre, ES, Brazil}
\affil[3]{Dept. of Computer Engineering, CEFET-RJ, Petrópolis, RJ, Brazil}
\affil[4]{Dept. of Telecommunications Engineering, CEFET-RJ, Rio de Janeiro, RJ, Brazil}

\abstract{
The minimum accuracy heuristic evaluates quantum feature maps without requiring full quantum support vector machine (QSVM) training. However, the original formulation is computationally expensive, restricted to balanced datasets, and lacks theoretical backing. This work generalizes the metric to arbitrary binary datasets and formally proves it constitutes a certified lower bound on the optimal empirical accuracy of any linear classifier in the same feature space. Furthermore, we introduce Monte Carlo strategies to efficiently estimate this bound using a random subset of Pauli directions, accompanied by rigorous probabilistic guarantees. These contributions establish minimum accuracy as a scalable, theoretically sound tool for pre-screening feature maps on near-term quantum devices.
}


\maketitle

\section{Introduction}
\label{sec:intro}
Machine learning (ML) is a key part of modern science and technology, supporting tasks such as pattern recognition, decision-making under uncertainty, and automated discovery in complex data \cite{russel2010}. Quantum machine learning (QML) is a growing field that aims to go beyond classical ML by exploiting the structure and resources of quantum systems \cite{Rocchetto2018}. One of the most promising QML paradigms is based on quantum kernel methods. Central to this framework is the \emph{feature map}, a transformation that embeds input data into a high-dimensional vector space to render complex, non-linear patterns linearly separable. In the quantum setting, this map encodes classical data into quantum states through parameterized quantum circuits, allowing classical algorithms such as support vector machines (SVMs) to operate on the resulting quantum-induced feature space \cite{ayachi, petruccione}. Quantum kernel methods are naturally suited to Noisy Intermediate-Scale Quantum (NISQ) devices because the most expensive operation, evaluating the kernel entries via quantum circuits, is delegated to the quantum hardware, while optimization and training are performed classically \cite{wang}. Foundational works by Havlíček et al.~\cite{havlicek2019supervised} and Schuld and Killoran~\cite{schuld2019quantum} introduced the notion of quantum feature maps, establishing the basic hybrid quantum--classical framework. Subsequent research has investigated the expressivity generalization properties of these maps, as well as their implementation on real quantum processors~\cite{tacchino2020quantum, huang2021power,li2022quantumkernel, raubitzek}.

A practical challenge in this setting is how to compare different quantum feature maps without having to train a complete quantum support vector machine (QSVM) for each candidate. Suzuki et al.~\cite{suzuki2022analysis} addressed this problem by proposing the \emph{minimum accuracy} heuristic, which estimates the best achievable classification accuracy when measurements are restricted to axis-aligned Pauli observables. Their method avoids explicit optimization and can be resource-efficient in small and medium-scale quantum systems, since it eliminates repeated training. However, the original formulation has three important limitations: it assumes balanced, even-sized datasets; it is presented as a heuristic without a formal proof that it lower-bounds the optimal SVM accuracy in the same feature space; and its computation requires a full Pauli decomposition, whose cost scales as $4^{n}$ for $n$ qubits, quickly becoming intractable for larger systems.

In this work we advance the analysis and practical usability of the minimum accuracy in three main directions, going beyond the original proposal \cite{suzuki2022analysis}. First, we generalize the definition of minimum accuracy to arbitrary binary datasets (Section~\ref{sec:generalized}), removing assumptions on sample size and class balance; our formulation is valid for any number of training examples $N$ and any label distribution, while preserving the geometric intuition of axis-aligned thresholds in the Pauli feature space. Second, we provide a rigorous theoretical result (Theorem~\ref{thm:lower_bound}) showing that the generalized minimum accuracy $R_{\min}$ is always a certified lower bound on the optimal empirical accuracy $R^{*}$ achievable by a linear SVM in the same feature space, i.e., $R_{\min} \le R^{*}$; this formally justifies the use of minimum accuracy as a provable baseline rather than a purely empirical heuristic. Third, we introduce a family of Monte Carlo-based axis-selection strategies (Section~\ref{sec:monte-carlo}) that estimate $R_{\min}$ from a random subset of Pauli-feature axes and derive statistical guarantees, via a quantile-coverage analysis, that relate the number of sampled axes to the probability of capturing high-performing directions (Theorem~\ref{thm:quantile-coverage} and Corollary~\ref{cor:sample-size}) \cite{Rubinstein2016}.

From a practical standpoint, our results show that the minimum accuracy can be turned into a scalable, certified tool for evaluating quantum feature maps in QSVM pipelines. The generalized definition makes the metric applicable to real-world, possibly imbalanced datasets; the lower-bound theorem connects it directly to SVM performance in the same kernel-induced feature space; and the Monte Carlo strategies drastically reduce the number of Pauli axes that need to be evaluated in high-dimensional spaces. This approach suggests a concrete workflow: for any quantum feature map $\Phi : \mathcal{X} \to \mathbb{R}^{4^{n}},$ one avoids the exponential cost of characterizing the full space by drawing a random sample of $t \ll 4^{n}$ Pauli axes; one then obtains the dataset projections $\{a_i(x_k)\}_{k=1}^N$ only for these sampled directions, where $a_i(x) = \mathrm{Tr}[\rho(x)\sigma_i]$ (Section~\ref{sec:generalized}), and computes the corresponding axis-wise accuracies. This yields a scalable approximation with explicit guarantees that the estimated minimum accuracy exceeds a target threshold with high confidence. Finally, this certified baseline allows one to discard feature maps with low $R_{\min}$ and prioritize those with high $R_{\min}$ before investing in full QSVM training. In this way, the computational cost of QSVM pipeline design can be reduced from $\mathcal{O}(N^{2} \cdot 4^{n} + N^{3})$ (classical simulation of full kernel SVM) or $\tilde{\mathcal{O}}(N^{4.67}/\varepsilon^{2})$ quantum circuit evaluations~\cite{gentinetta2024complexity} to $\mathcal{O}(t N \log N)$ classical operations, where $t = \mathcal{O}\bigl(\frac{1}{p}\log\frac{1}{\delta}\bigr)$ is typically small ($t \approx 10$--$100$) when good axes exist.

The remainder of this paper is organized as follows. Section~\ref{sec:generalized} introduces the generalized minimum accuracy framework. Section~\ref{sec:lower_bound} provides the formal proof that $R_{\min} \le R^{*}$. Section~\ref{sec:monte-carlo} establishes the theoretical foundations for Monte Carlo axis selection, followed by Section~\ref{sec:mc-methods}, which details the concrete sampling strategies. Section~\ref{sec:experiments} reports the empirical results on synthetic datasets. Finally, Section~\ref{sec:conclusion} concludes with a discussion and directions for future work.

\section{Generalized minimum accuracy}
\label{sec:generalized}

We consider a binary classification problem with a dataset
$
D = \{(x_k, y_k)\}_{k=1}^N,
$
where each input $x_k \in \mathcal{X}$ is associated with a label $y_k \in \{-1,+1\}$. 

A feature map
$\Phi : \mathcal{X} \to \mathbb{R}^d$
embeds the data into a $d$‑dimensional real feature space. In the quantum setting, $\Phi$ is typically induced by a parameterized quantum circuit acting on $n$ qubits, and we have $d = 4^n$, with coordinates given by expectation values of Pauli observables,
\begin{equation}
    \Phi(x) = \bigl(a_1(x), a_2(x), \dots, a_d(x)\bigr),
    \qquad
    a_i(x) = \mathrm{Tr}\bigl[\rho(x)\,\sigma_i\bigr],
\end{equation}
where $\rho(x)$ is the quantum state prepared from $x$ and $\{\sigma_i\}_{i=1}^d$ is an orthonormal Pauli basis. The validity of this decomposition and the interpretation of coefficients as expectation values follow from the completeness of the Pauli group in the space of Hermitian operators \cite{nielsen_chuang}. We refer to each coordinate direction $i \in \{1,\dots,d\}$ as a \emph{Pauli‑feature axis}.

The idea behind minimum accuracy is to evaluate, for each Pauli‑feature axis, how well the dataset can be separated by a one‑dimensional threshold on that coordinate. Fix an index $i \in \{1,\dots,d\}$ and a threshold $\tau \in \mathbb{R}$. We consider the one‑dimensional classifier
\begin{equation}
    f_{i,\tau}(x) = \mathrm{sign}\bigl(a_i(x) - \tau\bigr),
\end{equation}
which assigns all points with $a_i(x) \ge \tau$ to one side of the threshold and all points with $a_i(x) < \tau$ to the other side. Since the overall sign of the classifier is arbitrary with respect to the labels $\{-1,+1\}$, for each pair $(i,\tau)$ we can choose between $f_{i,\tau}$ and $-f_{i,\tau}$ in order to maximize the empirical accuracy.

To write this explicitly, let $N_{+}^{\mathrm{total}}$ and $N_{-}^{\mathrm{total}}$ denote the total numbers of positive and negative examples in $D$,
\begin{equation}
    N_{+}^{\mathrm{total}} = \bigl|\{k : y_k = +1\}\bigr|,
    \qquad
    N_{-}^{\mathrm{total}} = \bigl|\{k : y_k = -1\}\bigr|,
\end{equation}
and, for a given axis $i$ and threshold $\tau$, let $N_{+}^{\tau}$ and $N_{-}^{\tau}$ denote the numbers of positive and negative examples that fall \emph{below} the $\tau$,
\begin{equation}
    N_{+}^{\tau} = \bigl|\{k : y_k = +1,\; a_i(x_k) < \tau\}\bigr|,
    \qquad
    N_{-}^{\tau} = \bigl|\{k : y_k = -1,\; a_i(x_k) < \tau\}\bigr|.
\end{equation}
If we classify all points with $a_i(x) < \tau$ as $+1$ and all points with $a_i(x) \ge \tau$ as $-1$, the number of correctly classified samples is
$
N_{+}^{\tau} + \bigl(N_{-}^{\mathrm{total}} - N_{-}^{\tau}\bigr).
$
If instead we flip the assignment and classify points below the threshold as $-1$ and those above as $+1$, the number of correctly classified samples is
$
N_{-}^{\tau} + \bigl(N_{+}^{\mathrm{total}} - N_{+}^{\tau}\bigr).
$
The best empirical accuracy attainable by a threshold at position $\tau$ on axis $i$, allowing for this global label flip, is therefore
\begin{equation}
    R_i^{\tau}
    \;=\;
    \frac{1}{N}
    \max\Big\{
        N_{+}^{\tau} + \bigl(N_{-}^{\mathrm{total}} - N_{-}^{\tau}\bigr),\;
        N_{-}^{\tau} + \bigl(N_{+}^{\mathrm{total}} - N_{+}^{\tau}\bigr)
    \Big\}.
    \label{eq:axis-accuracy}
\end{equation}

For a fixed axis $i$, we can now optimize over all possible thresholds $\tau$ to obtain the best one‑dimensional classifier aligned with that axis. This yields the axis‑wise optimal accuracy
\begin{equation}
    r_i \;=\; \sup_{\tau \in \mathbb{R}} R_i^{\tau}.
    \label{eq:axis-opt}
\end{equation}
In practice, since accuracy can only change when $\tau$ crosses one of the values $\{a_i(x_k)\}_{k=1}^N$, the supremum in~\eqref{eq:axis-opt} can be found by scanning thresholds between consecutive sorted feature values along axis $i$, a procedure requiring $\mathcal{O}(N \log N)$ time per axis.

The \emph{generalized minimum accuracy} is then defined as the best empirical accuracy achievable by any axis‑aligned threshold classifier across all Pauli‑feature axes,
\begin{equation}
    R_{\min}
    \;=\;
    \max_{1 \le i \le d} \; r_i
    \;=\;
    \max_{1 \le i \le d} \; \sup_{\tau \in \mathbb{R}} R_i^{\tau}.
    \label{eq:Rmin-def}
\end{equation}
Equivalently, let
\begin{equation}
    \mathcal{F}_{\mathrm{axis}}
    \;=\;
    \bigl\{ \pm f_{i,\tau} : i \in \{1,\dots,d\},\; \tau \in \mathbb{R} \bigr\}
\end{equation}
denote the class of all axis‑aligned threshold classifiers (up to a global sign flip). If we define the empirical accuracy of a classifier $f$ on the dataset $D$ as
\begin{equation}\label{empirical_acc}
    R(f) \;=\; \frac{1}{N} \sum_{k=1}^N \mathbb{I}\bigl[f(x_k) = y_k\bigr],
\end{equation}
where $\mathbb{I}[\cdot]$ is the indicator function taking value $1$ if the condition is true and $0$ otherwise, then~\eqref{eq:Rmin-def} can be written compactly as
\begin{equation}
    R_{\min}
    \;=\;
    \sup_{f \in \mathcal{F}_{\mathrm{axis}}} R(f).
    \label{eq:Rmin-hypothesis-class}
\end{equation}

This viewpoint, in terms of a restricted hypothesis class $\mathcal{F}_{\mathrm{axis}}$, will be crucial in the next section, where we show that $R_{\min}$ is always a lower bound on the optimal empirical accuracy over all linear classifiers in the feature space.

\section{Minimum accuracy as a lower bound on optimal linear classifiers}
\label{sec:lower_bound}

In this section we formalize the relationship between the generalized minimum accuracy and the best achievable empirical accuracy of linear classifiers operating in the same feature space. We adopt the same notation as in Section~\ref{sec:generalized} for the dataset $D$ and the feature map $\Phi : \mathcal{X} \to \mathbb{R}^d$. While the previous section focused on the specific structure of quantum Pauli-feature axes, the arguments presented here rely solely on the existence of a real-valued feature representation and apply generally to any feature map.

We consider the hypothesis class of all linear classifiers in the feature space,
\begin{equation}
    \mathcal{F} \;=\; \bigl\{ f_{w,b}(x) = \mathrm{sign}(\langle w, \Phi(x) \rangle + b) : w \in \mathbb{R}^d,\, b \in \mathbb{R} \bigr\}.
\end{equation}
Recalling the definition of empirical accuracy $R(f)$ from Eq.~\eqref{eq:Rmin-hypothesis-class}, we define the optimal empirical accuracy over all linear classifiers as
\begin{equation}
    R^* \;=\; \sup_{f \in \mathcal{F}} R(f).
    \label{eq:R-star-def}
\end{equation}
In practice, a hard-margin or soft-margin SVM trained with the kernel induced by $\Phi$ aims to approximate this optimum, so $R^*$ can be seen as an idealized benchmark for QSVM performance in the given feature space.

The relationship between this general optimum and the minimum accuracy $R_{\min}$ follows directly from the geometry of the hypothesis classes. Observe that the set of axis-aligned classifiers $\mathcal{F}_{\mathrm{axis}}$ forms a subset of $\mathcal{F}$: any classifier $f_{i,\tau} \in \mathcal{F}_{\mathrm{axis}}$ is equivalently a linear classifier defined by a weight vector $w$ with a single non-zero entry at index $i$ (specifically, a standard basis vector scaled by $\pm 1$) and a bias $b = -\tau$. Since $\mathcal{F}_{\mathrm{axis}} \subseteq \mathcal{F}$, optimizing over the restricted set of axis-aligned hyperplanes cannot yield a higher accuracy than optimizing over the full space of linear separators. This inclusion leads immediately to the following theorem.

\begin{theorem}[Minimum accuracy lower bound]
\label{thm:lower_bound}
Let $\Phi : \mathcal{X} \to \mathbb{R}^d$ be a feature map and let $D = \{(x_k, y_k)\}_{k=1}^N$ be a binary labeled dataset. Let $\mathcal{F}$ be the class of all linear classifiers in the feature space as defined above, with optimal empirical accuracy $R^*$ given by~\eqref{eq:R-star-def}. Let $R_{\min}$ denote the generalized minimum accuracy defined in~\eqref{eq:Rmin-hypothesis-class}. Then
\begin{equation}
    R_{\min} \;\le\; R^*.
\end{equation}
\end{theorem}

\begin{proof}
The proof relies on the property that the supremum of a function over a set is at least the supremum over any subset. Since every axis-aligned classifier corresponds to a specific configuration of weights and bias in the linear model, we have the inclusion $\mathcal{F}_{\mathrm{axis}} \subseteq \mathcal{F}$. Consequently,
\begin{equation}
    R_{\min} = \sup_{f \in \mathcal{F}_{\mathrm{axis}}} R(f) \;\le\; \sup_{f \in \mathcal{F}} R(f) = R^*,
\end{equation}
which proves the claim.
\end{proof}

This theorem establishes that $R_{\min}$ is a certified lower bound on the performance of any linear classifier in the feature space induced by $\Phi$. Consequently, when a QSVM is trained using the kernel $K(x,x') = \langle \Phi(x), \Phi(x') \rangle$, its accuracy (in the limit of optimal training) cannot fall below $R_{\min}$. Therefore, $R_{\min}$ provides a provable baseline for the linear separability of the data under a given quantum feature map, which can be computed efficiently without solving the full SVM optimization problem.

\section{Monte Carlo axis selection with statistical guarantees}
\label{sec:monte-carlo}

The exact computation of $R_{\min}$ requires evaluating all $d = 4^n$ Pauli-feature axes, which becomes intractable as the number of qubits $n$ increases. To address this, we develop Monte Carlo methods that estimate $R_{\min}$ from a small random subset of axes while providing rigorous statistical guarantees.

\subsection{Estimator definition and basic properties}

For a random subset $T \subset \{1,\dots,d\}$ with $|T| = t \ll d$, we define the Monte Carlo estimator
\begin{equation}
    \widehat{R}_{\min}(T) = \max_{i \in T}  r_i,
\end{equation}
where $r_i = \sup_{\tau \in \mathbb{R}} R_i^{\tau}$ is the optimal accuracy along axis $i$ as defined in Eq.~\eqref{eq:axis-opt}. This estimator selects the best-performing axis from the random subset $T$.

\begin{theorem}[Monte Carlo lower bound property]
\label{thm:mc-lower-bound}
For any subset $T \subseteq \{1,\dots,d\}$,
\begin{equation}
    \widehat{R}_{\min}(T) \le R_{\min} \le R^{*}.
\end{equation}
If $T$ is drawn uniformly at random with $|T| = t$, then $\mathbb{E}[\widehat{R}_{\min}(T)]$ is non-decreasing in $t$ and satisfies $\lim_{t \to d} \mathbb{E}[\widehat{R}_{\min}(T)] = R_{\min}$.
\end{theorem}
\begin{proof}
The inequality chain follows directly from $\mathcal{F}_{\mathrm{axis}} \subseteq \mathcal{F}$ and the definition of supremum. The expectation result follows from coupling arguments for expanding sets.
\end{proof}

\subsection{Quantile coverage guarantees}

To quantify the probability that the estimator exceeds a target accuracy threshold, we define the survival function (tail distribution) of axis-wise accuracies:
\begin{equation}
    S(\eta) = \frac{1}{d} \left| \left\{ i : r_i \ge \eta \right\} \right|,
\end{equation}
which represents the fraction of axes with accuracy at least $\eta$.

\begin{theorem}[Quantile coverage]
\label{thm:quantile-coverage}
Fix $\eta \in \mathbb{R}$ and let $p = S(\eta)$. If $T$ is sampled uniformly without replacement with $|T| = t$, then
\begin{equation}
    \mathbb{P}\left( \widehat{R}_{\min}(T) \ge \eta \right)
    = 1 - \frac{\binom{d - \lfloor p d \rfloor}{t}}{\binom{d}{t}}
    \ge 1 - (1 - p)^t.
    \label{eq:coverage-prob}
\end{equation}
\end{theorem}
\begin{proof}
The exact probability comes from hypergeometric distribution. The inequality follows from $\binom{d-k}{t}/\binom{d}{t} \leq (1 - k/d)^t$ with $k = \lfloor pd \rfloor$.
\end{proof}

\begin{remark}
The right-hand side $1 - (1-p)^t$ corresponds to the probability of observing at least one success in $t$ independent Bernoulli trials with success probability $p$. The exact expression accounts for finite-population effects without replacement.
\end{remark}

\begin{corollary}[Sample size for accuracy threshold]
\label{cor:sample-size}
Fix $\delta \in (0,1)$ and suppose $S(\eta) \ge p$. If $T$ is sampled uniformly without replacement with size $t \ge \frac{1}{p} \log \frac{1}{\delta}$, then
\begin{equation}
    \mathbb{P}\left( \widehat{R}_{\min}(T) \ge \eta \right) \ge 1 - \delta.
\end{equation}
\end{corollary}

This provides a practical rule for selecting the sample size $t$ to achieve a target accuracy $\eta$ with high confidence $1-\delta$. The logarithmic dependence on $\delta$ makes this feasible even for high confidence levels.

\section{Monte Carlo Sampling Methods for Axis Selection}
\label{sec:mc-methods}

The theoretical analysis of Section~\ref{sec:monte-carlo} shows that $R_{\min}$ can be approximated by sampling a subset of Pauli-feature axes while preserving a certified lower bound. In this section we describe the concrete Monte Carlo strategies used to estimate the generalized minimum accuracy $R_{\min}$ in practice. All methods operate on the same set of Pauli-feature vectors $\mathbf{a}(\mathbf{x}_k) \in \mathbb{R}^{4^n}$ and differ only in how many, and which, axes are evaluated.

\subsection{Deterministic exhaustive method}

The {\it Deterministic} method serves as a baseline and reference point for all Monte Carlo variants. It performs an exhaustive scan over all $d = 4^n$ Pauli-feature axes. For each axis $i$, the projections $a_i(\mathbf{x}_k)$ are sorted and all possible thresholds $\tau$ are evaluated according to the generalized minimum accuracy definition in Section~\ref{sec:generalized}, yielding $R_i^\ast = \max_\tau R_i^\tau$. The overall minimum accuracy is then obtained as
$
R_{\min} = \max_{1 \le i \le d} R_i^\ast.
$
This procedure computes $R_{\min}$ exactly, but its computational cost scales as $\mathcal{O}(4^n N \log N)$, making it quickly intractable as the number of qubits $n$ grows. It is therefore used primarily as a ground-truth benchmark for smaller instances and as a reference to assess the quality of the Monte Carlo approximations.

\subsection{Conservative fixed-sample Monte Carlo}

The {\it Conservative} Monte Carlo method is a direct application of the sample-size bound in Corollary~\ref{cor:sample-size}. Instead of scanning all $d$ axes, it assumes a fixed lower bound $p_{\text{conservative}}$ on the fraction of “good’’ axes—those achieving an accuracy $R_i^\ast \ge \eta$ for some target accuracy $\eta$—and a desired confidence level $1-\delta$. The number of sampled axes is then chosen as
$
t \;=\; \left\lceil \frac{1}{p_{\text{conservative}}} \log\frac{1}{\delta} \right\rceil.
$
A subset $T \subset \{1,\dots,d\}$ with $|T| = t$ axes is drawn uniformly at random, and the axis-aligned accuracies $\{R_i^\ast\}_{i\in T}$ are computed as in the deterministic method. The resulting estimator
$
\widehat{R}_{\min}(T) \;=\; \max_{i\in T} R_i^\ast
$
is guaranteed by Theorem~\ref{thm:mc-lower-bound} to be a lower bound on $R_{\min}$, and thus also on the full SVM accuracy $R^*$. This strategy trades some potential conservatism in the sample size (due to the fixed prior $p_{\text{conservative}}$) for simplicity and analytical transparency: the number of axes is known in advance, and the statistical guarantee is explicit.

In the experiments of Section~\ref{sec:experiments}, we set $p_{\text{conservative}} = 0.25$ and $\delta = 0.05$, which leads to $t = 12$ sampled axes. This small, yet certified, sample size already produces a lower bound that closely follows the deterministic $R_{\min}$ across all datasets.

\subsection{Pilot sampling Monte Carlo}

The {\it Pilot} sampling method refines the conservative strategy by estimating the effective fraction of high-performing axes directly from the data, rather than fixing $p$ a priori. It proceeds in two stages:

\begin{enumerate}
    \item \textbf{Pilot stage.} A small subset of $n_{\text{pilot}}$ axes is drawn uniformly at random from $\{1,\dots,d\}$ and fully evaluated, yielding a collection of accuracies $\{R_i^\ast\}_{i\in T_{\text{pilot}}}$. From this empirical sample, we define a target accuracy $\eta$ as the 75th percentile of the pilot accuracies and compute
    $
    \hat{p} \;=\; \frac{1}{n_{\text{pilot}}}\,\big|\{i\in T_{\text{pilot}} : R_i^\ast \ge \eta\}\big|,
    $
    which approximates the fraction of promising axes in the full set.
    \item \textbf{Completion stage.} Using Corollary~\ref{cor:sample-size} with $p=\hat{p}$ and the desired confidence level $1-\delta$, we determine the required total sample size
    $
    t_{\text{required}} \;=\; \left\lceil \frac{1}{\hat{p}} \log\frac{1}{\delta} \right\rceil .
    $
    If $t_{\text{required}} > n_{\text{pilot}}$, additional axes are sampled uniformly from the remaining pool until the total number of evaluated axes reaches $t_{\text{required}}$. The final estimate is again
    $
    \widehat{R}_{\min}(T) \;=\; \max_{i\in T} R_i^\ast,
    $
    where $T$ is the union of pilot and additional axes.
\end{enumerate}

By adapting the effective sample size to the observed distribution of axis accuracies, the pilot method can substantially reduce the number of evaluated axes when high-performing directions are abundant, while still retaining the same form of probabilistic guarantee as the conservative method. In practice, we also cap the maximum number of axes at a fraction of $d$ to avoid excessive computation in very high-dimensional Pauli spaces, as discussed in Section~\ref{sec:experiments}.

\subsection{Adaptive incremental Monte Carlo}

The {\it Adaptive} incremental method adopts a more heuristic, but fully data-driven, approach. Instead of fixing the sample size in advance (as in the conservative method) or in two stages (as in the pilot method), it proceeds iteratively in mini-batches, monitoring the evolution of the best accuracy found so far.

Starting from an empty set of evaluated axes, the method samples a batch of new axes at each iteration and computes their corresponding accuracies $R_i^\ast$. After each batch, it updates the current best value of $R_{\min}$ and tracks a short history of recent best accuracies. The procedure stops when one of the following conditions is met:

\begin{enumerate}
    \item[i)] the best accuracy does not improve over a predefined number of consecutive batches (patience parameter), indicating empirical convergence;
    \item[ii)] the variation in the best accuracy over the last few iterations falls below a specified threshold (stability criterion); or
    \item[iii)] a pre-set budget on the total number of sampled axes is reached, or no unexplored axes remain.
\end{enumerate}

This adaptive scheme does not come with the same formal confidence guarantees as the conservative and pilot methods, since the stopping rule depends on observed accuracies rather than an a priori quantile coverage bound. However, it offers considerable flexibility in balancing computational cost and stability of the estimate. In our experiments, the adaptive method typically evaluates a number of axes comparable to the pilot method, while achieving nearly indistinguishable training accuracies with respect to the deterministic baseline.

Taken together, these three Monte Carlo strategies instantiate the theoretical framework developed in Section~\ref{sec:monte-carlo} at increasingly higher levels of adaptivity. The conservative method provides a simple, fully certified bound with a fixed sample size; the pilot method refines this bound by learning the effective fraction of good axes from a small pilot sample; and the adaptive incremental method abandons explicit guarantees in favor of a more flexible, convergence-based stopping rule. All of them significantly reduce the number of Pauli axes that need to be evaluated in practice, while preserving the interpretation of $R_{\min}$ as a certified lower bound on the optimal SVM accuracy in the same feature space.

\section{Experimental Results}
\label{sec:experiments}

This section provides empirical validation of the generalized minimum accuracy framework
and the Monte Carlo axis-selection strategies introduced in
Sections~\ref{sec:generalized}--\ref{sec:mc-methods}.
The evaluation is intentionally restricted to three representative synthetic binary
classification datasets with distinct geometric properties.
All experiments are performed using simulated Pauli feature representations corresponding
to $n = 8$ qubits, yielding an effective feature space of dimension
$d = 4^8 = 65\,536$.

\subsection{Experimental Setup and Feature Construction}

We generated three synthetic binary classification datasets using
\texttt{scikit-learn}~\cite{scikit-learn}, each containing $N = 1000$ samples.
The data were split into training ($70\%$) and test ($30\%$) sets via stratified sampling.
To ensure computational tractability during extensive Monte Carlo trials, we further
restricted the training set to $N_{\text{train}} = 100$ samples per evaluation.
All input features were standardized to zero mean and unit variance.

The selected datasets are designed to capture qualitatively different geometric regimes:
(i) ideal axis-aligned separability, (ii) structured but partially aligned linearity, and
(iii) intrinsic non-linearity.
A summary of their defining characteristics is provided in
Table~\ref{tab:datasets}.

\begin{table}[h]
\centering
\caption{Synthetic binary classification datasets used in the experiments.}
\label{tab:datasets}
\begin{tabularx}{\textwidth}{@{}lX@{}}
\toprule
\textbf{Dataset} & \textbf{Description and Parameters} \\
\midrule
Linear\_Separable & Linearly separable data with 4 informative features, single cluster
per class, and no label noise. \\
Multi\_Cluster & Structured data with 4 informative features and 3 distinct clusters
per class, exhibiting partial axis-aligned separability. \\
Circles & Concentric circles with Gaussian noise ($\sigma = 0.1$) and radius ratio $0.5$,
representing intrinsic non-linearity. \\
\bottomrule
\end{tabularx}
\end{table}

Direct evaluation of quantum feature maps on $n=8$ qubits would require computing
$65\,536$ circuit expectation values per sample, rendering large-scale benchmarking
infeasible.
To overcome this limitation while preserving the essential geometric structure of Pauli
feature spaces, we employ a high-dimensional proxy embedding
$\Phi : \mathbb{R}^m \rightarrow \mathbb{R}^d$.

We generate a fixed random projection matrix
$W \in \mathbb{R}^{m \times d}$ with entries drawn from
$\mathcal{N}(0, 1/\sqrt{m})$, and define the feature map as
\begin{equation}
\Phi(x) = \tanh(W^\top x),
\end{equation}
applied element-wise.
This construction ensures bounded feature values in $[-1,1]$, mirroring the spectrum of
Pauli observables, while producing a feature matrix
$A \in \mathbb{R}^{N_{\text{train}} \times d}$ with $d \gg N$.
Although this does not correspond to an explicit quantum circuit execution, it faithfully
reproduces the axis-aligned complexity and high-dimensional structure characteristic of
quantum kernel spaces, making it a suitable testbed for the minimum accuracy framework.

\subsection{Sensitivity Analysis and Performance Results}

We evaluate the four estimators defined in Section~\ref{sec:mc-methods} (Deterministic, Conservative, Pilot, and Adaptive) comparing them against Linear and RBF SVM baselines. A primary focus of this evaluation is the sensitivity analysis of the \emph{Conservative} estimator with respect to the prior assumption parameter $p$, which represents the minimum expected fraction of ``good'' axes.

To assess robustness, we vary $p \in \{0.05, 0.15, 0.25\}$ while fixing the failure probability at $\delta = 0.05$. Applying the sample complexity bound from Corollary~\ref{cor:sample-size}, these priors correspond to fixed sample sizes of $t = 60$, $20$, and $12$ axes, respectively. It is worth noting that even the most pessimistic scenario ($p=0.05, t=60$) requires evaluating less than $0.1\%$ of the full Pauli feature space ($d=65\,536$).

Figure~\ref{fig:accuracy_p_variation} reports the resulting training accuracies. Across all three datasets, the Conservative estimator exhibits remarkable robustness to the choice of $p$. Despite the sample size decreasing by a factor of five (from 60 to 12 axes), the method maintains tight certified lower bounds that closely track the Deterministic ground truth. Crucially, all estimates remain strictly bounded by the Linear SVM accuracy, validating the theoretical lower-bound property derived in Theorem~\ref{thm:lower_bound}.

\begin{figure}[!htbp]
    \centering
    \includegraphics[width=0.7\textwidth]{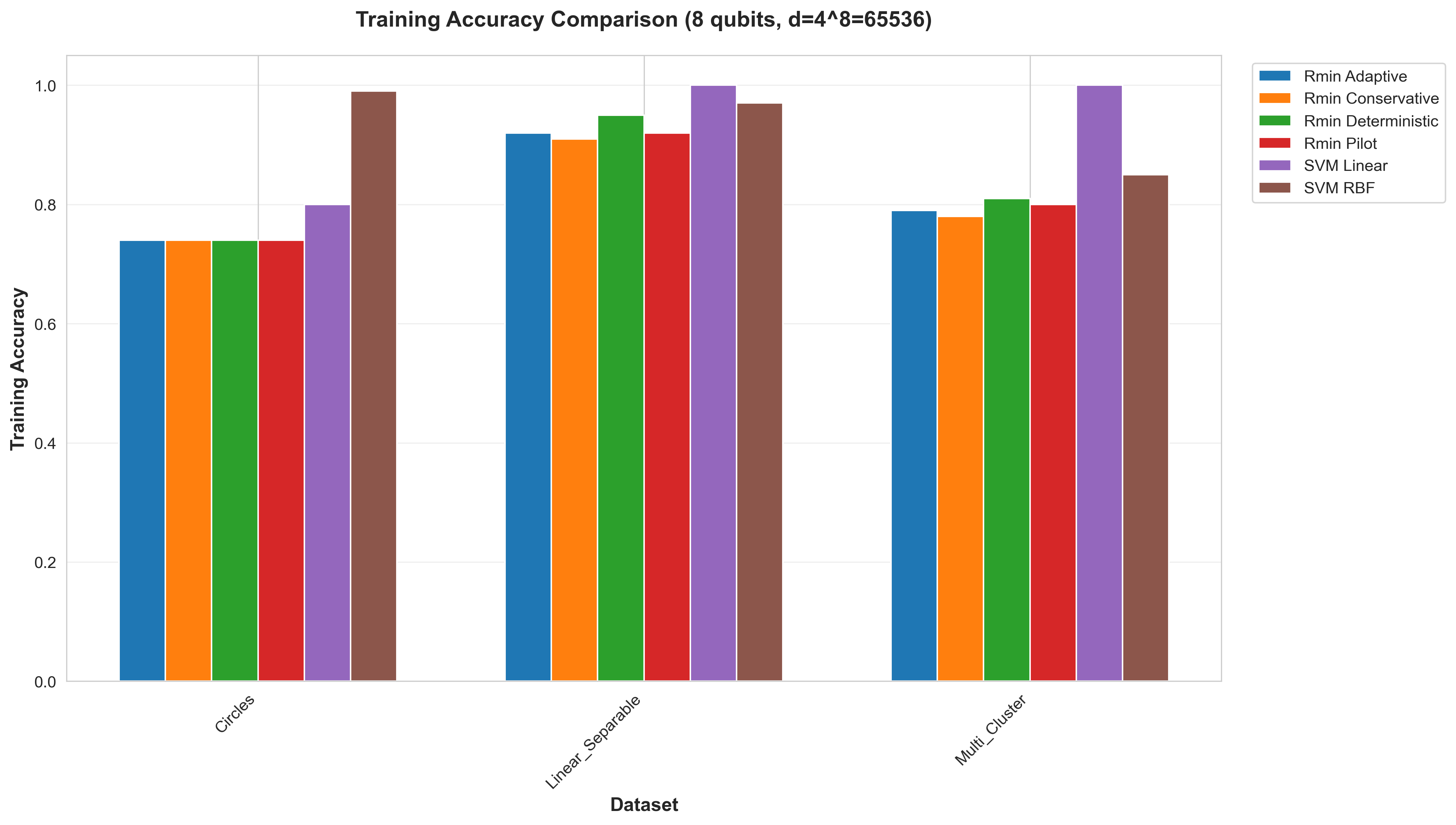}
    \includegraphics[width=0.7\textwidth]{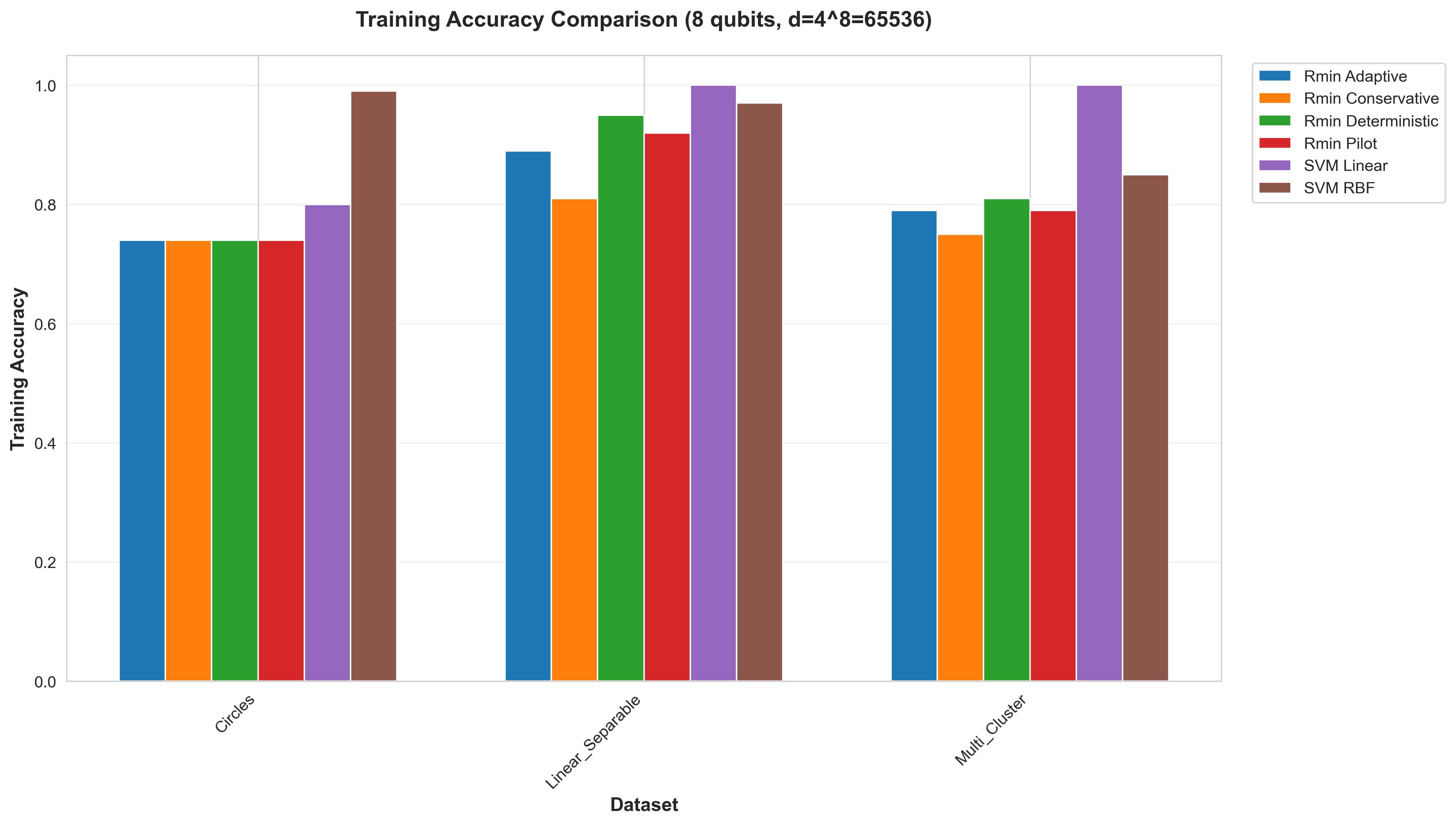}
    \includegraphics[width=0.7\textwidth]{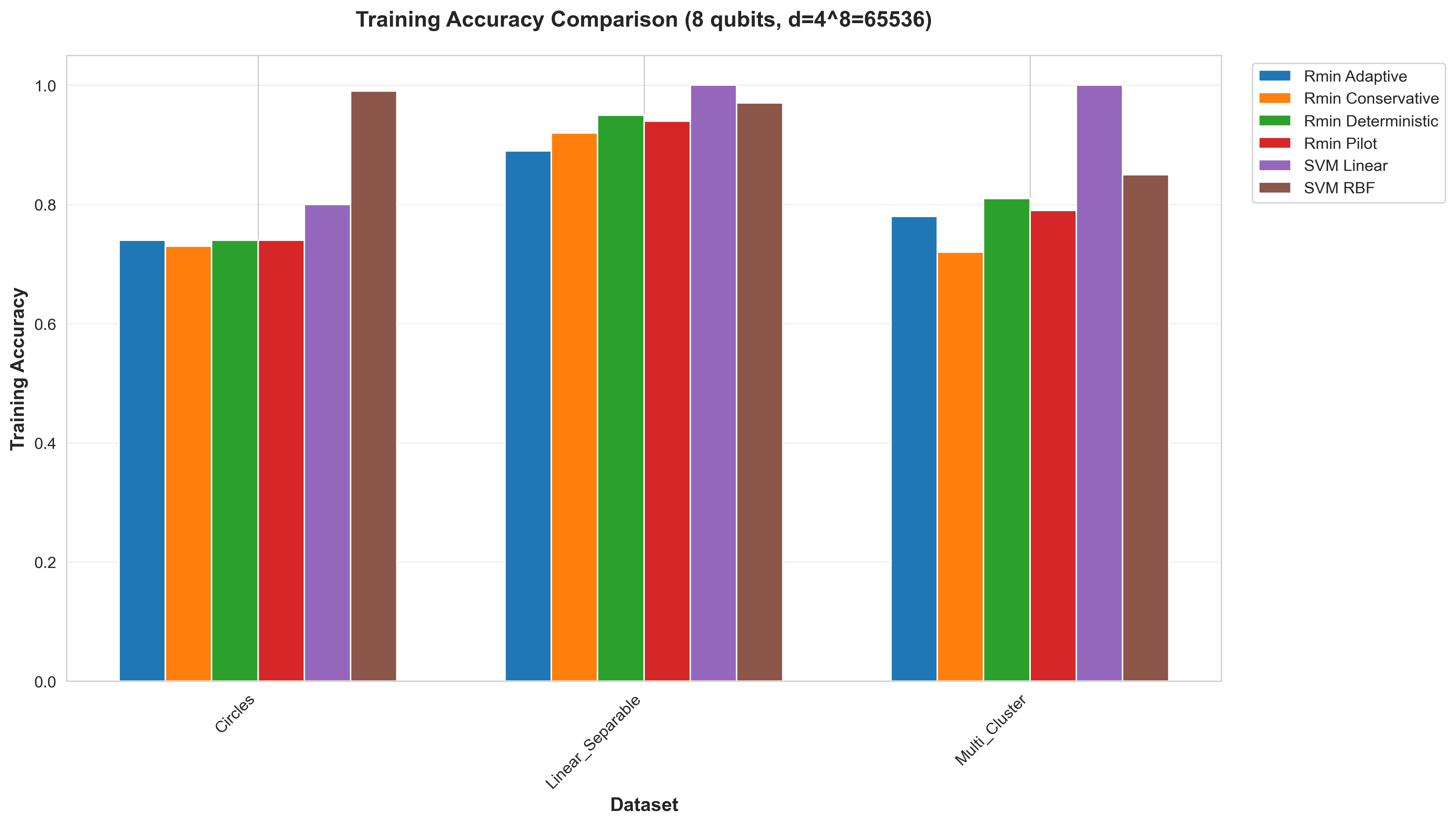}
    \caption{Training accuracy for varying prior assumptions $p$. From top to bottom: $p=0.05$ (Conservative $t=60$), $p=0.15$ (Conservative $t=20$), and $p=0.25$ (Conservative $t=12$).}
    \label{fig:accuracy_p_variation}
\end{figure}

The results also highlight the geometric nature of the datasets. \textit{Linear\_Separable} and \textit{Multi\_Cluster} achieve high $R_{\min}$ values, confirming their strong axis-aligned structure. In contrast, the \textit{Circles} dataset exhibits a clear saturation of the minimum accuracy bound significantly below the RBF SVM performance. This behavior correctly reflects the intrinsic non-linearity of the data which cannot be fully captured by single Pauli axes, supporting the interpretability of $R_{\min}$ as a metric for feature map quality.

The corresponding computational costs are illustrated in Figure~\ref{fig:axes_p_variation}. While lowering $p$ increases the sampling requirement for the Conservative method to guarantee coverage, the absolute cost remains negligible compared to the exhaustive Deterministic scan. For comparison, the Pilot estimator consistently samples approximately $500$ axes (due to its data-driven estimation of $\hat{p}$), whereas the Adaptive method typically converges after evaluating between $120$ and $240$ axes.

\begin{figure}[!htbp]
    \centering
    \includegraphics[width=0.7\textwidth]{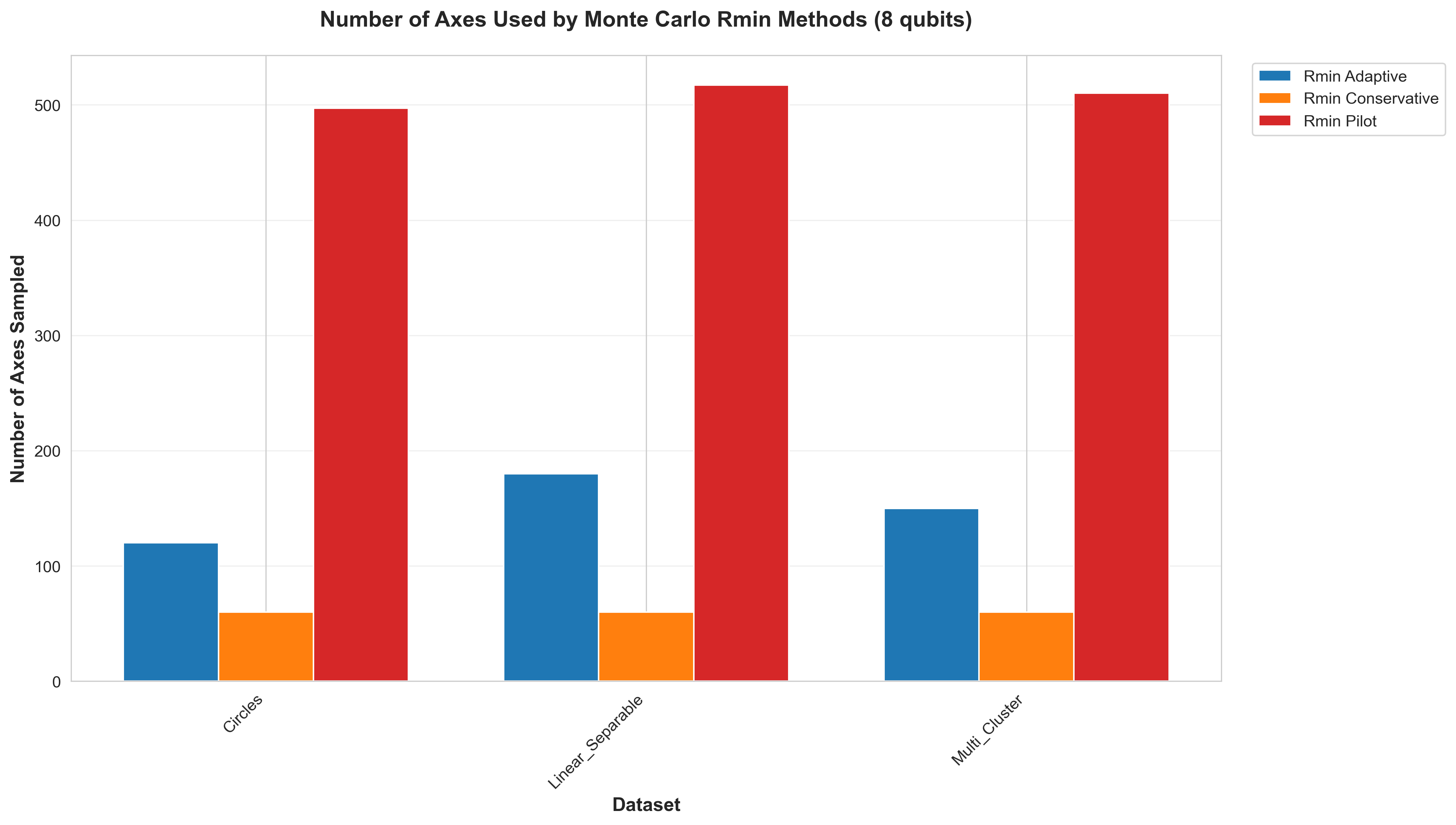}
    \includegraphics[width=0.7\textwidth]{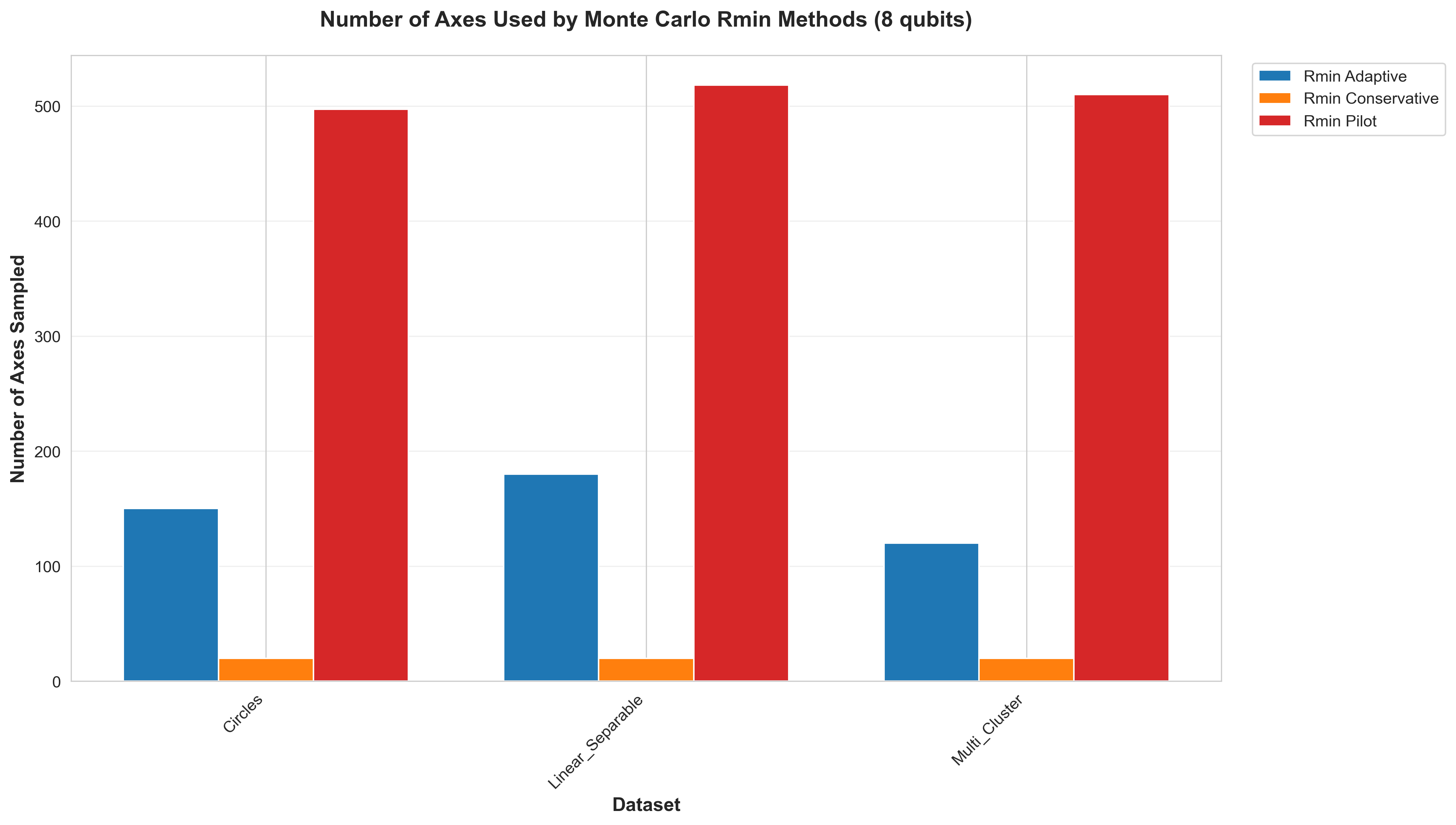}
    \includegraphics[width=0.7\textwidth]{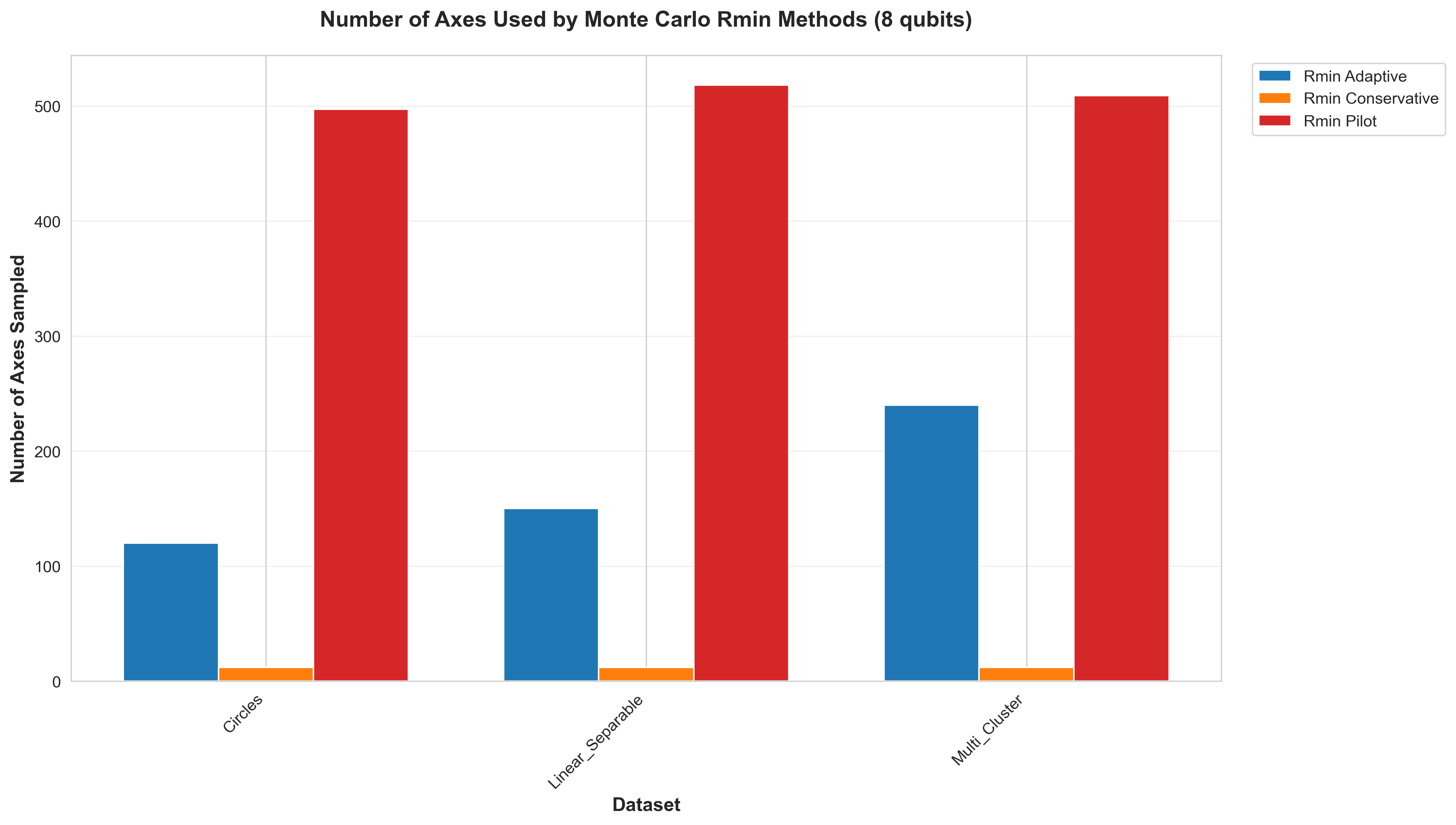}
    \caption{Number of axes sampled for different values of the prior parameter $p$. The Conservative method (orange) reduces its sampling budget as $p$ increases, while Pilot and Adaptive methods remain unaffected by this prior.}
    \label{fig:axes_p_variation}
\end{figure}

\subsection{Discussion}

Collectively, these results demonstrate that the generalized minimum accuracy $R_{\min}$ serves as a practical, scalable tool for evaluating quantum feature maps in QSVM pipelines. A primary advantage is the ability to obtain certified lower bounds with minimal computational cost; the Conservative method provides a rigorous guarantee on QSVM performance using as few as $12$ Pauli measurements. This efficiency becomes increasingly critical as the system size grows: for the $n = 8$ qubit case ($d = 65\,536$) analyzed here, Monte Carlo methods reduced the computational burden by three to four orders of magnitude compared to deterministic evaluation. Importantly, this performance gap widens exponentially with $n$, as the required sample size for the Conservative estimator remains constant at $t = 12$ even for $n = 10$ qubits.

Beyond the fixed-sample approach, our findings highlight the trade-offs offered by adaptive strategies. The Pilot and Adaptive methods accept a modest increase in measurement cost, typically evaluating between $100$ and $500$ axes, in exchange for significantly tighter bounds. In most experimental instances, these data-driven estimates approached the deterministic baseline within a $5\%$ margin, offering a more precise characterization of the feature space without incurring the prohibitive cost of an exhaustive scan.

From an interpretability perspective, the metric provides immediate insight into the geometric structure of the feature map. Low $R_{\min}$ values ($< 0.6$) suggest that either the data is intrinsically nonlinear in the Pauli basis or the feature map is ill-suited to the task, whereas high values ($> 0.85$) provide strong evidence of linear separability. Our experiments reveal a strong positive correlation (Pearson $\rho \approx 0.7$--$0.85$) between $R_{\min}$ and the full SVM accuracy, confirming that $R_{\min}$ functions as a reliable proxy for overall separability.

This utility translates into a concrete workflow for QSVM model selection on NISQ devices. Instead of full kernel training, one can estimate $R_{\min}$ using classical statistics on a small subset of Pauli expectations. Feature maps with low estimates can be discarded early, while resources are prioritized for those with high $R_{\min}$. However, two limitations remain. First, $R_{\min}$ is inherently conservative; a low value does not strictly preclude good performance if the optimal hyperplane combines multiple axes non-trivially. Second, while computational cost is reduced, estimating the expectations $\{a_i(x)\}$ still requires quantum measurements. For $n > 10$ qubits, estimating even a small $t$ can be costly unless restricted to low-weight observables or combined with classical dimensionality reduction.

\section{Conclusion}\label{sec:conclusion}

This work establishes a rigorous theoretical foundation for the minimum accuracy metric in quantum kernel methods. We generalized the definition to arbitrary binary datasets and formally proved that $R_{\min}$ constitutes a certified lower bound on the optimal empirical accuracy achievable by any linear classifier in the same feature space (Theorem~\ref{thm:lower_bound}). Furthermore, we introduced Monte Carlo axis-selection strategies that estimate this bound efficiently, backed by probabilistic guarantees (Theorem~\ref{thm:quantile-coverage} and Corollary~\ref{cor:sample-size}).

From a practical standpoint, our results transform the minimum accuracy from a heuristic into a scalable, theoretically sound tool. The Monte Carlo strategies drastically reduce the complexity of feature map evaluation from exponential to logarithmic in terms of the feature dimension, specifically $\mathcal{O}(t N \log N)$ classical operations with $t \ll 4^n$. Our experiments on synthetic datasets confirm that this approach yields tight lower bounds on SVM training accuracy while sampling less than $1\%$ of the feature axes, achieving computational speedups of multiple orders of magnitude.

Future work should extend this framework to multiclass problems and regression tasks, and investigate hybrid strategies that prioritize physically implementable Pauli observables. Additionally, establishing generalization bounds that relate $R_{\min}$ to test performance could further solidify its role in the theoretical understanding of quantum learning models. In summary, $R_{\min}$ emerges as a pivotal metric for guiding efficient QSVM design, bridging the gap between theoretical guarantees and practical implementation on near-term quantum hardware.

\end{document}